\documentclass[runningheads]{llncs}
\usepackage{graphicx}
\usepackage{hyperref,xcolor}
\hypersetup{
    colorlinks = true,
    linkcolor = {blue},
    citecolor = {blue},
}

\usepackage{mathtools}
\usepackage{listings}
\usepackage{proof}
\usepackage{amsfonts}
\usepackage{amssymb}
\usepackage{stmaryrd}

\begin{document}
\title{Extending Concurrent Separation Logic to Enhance Modular Formalization
% \thanks{Supported by organization x.}
}
% \author{Yepeng Ding\orcidID{1111-2222-3333-4444} \and
% Hiroyuki Sato\orcidID{1111-2222-3333-4444}}
\author{Yepeng Ding \and
Hiroyuki Sato}
\authorrunning{Y. Ding and H. Sato}
\institute{The University of Tokyo, Tokyo, Japan\\
\email{\{youhoutei,schuko\}@satolab.itc.u-tokyo.ac.jp}}
\maketitle              % typeset the header of the contribution
\begin{abstract}
Nowadays, numerous services based on large-scale distributed systems have been developed to boost the convenience of human life. On the other side, it becomes a significant challenge to ensure the correctness and properties of these systems due to the complex and nested architecture. Although concurrent separation logic (CSL) has partially tackled the problem by specifying systems and verifying the correctness of them, it faces modularity issues. In this paper, we propose an extended concurrent separation logic (ECSL) to address the modularity issues of CSL with the support of the temporal extension, communication extension, environment extension, and nest extension. ECSL is capable of formalizing systems at different abstraction levels from memory management to architecture and protocol design with great modularity. Furthermore, we stick to unitarity and compatibility principles while developing ECSL.

\end{abstract}
\section{Introduction}

Thanks to the advancement of distributed computing techniques, numerous large-scale systems have been developed to provide diverse services benefiting people in daily life. The architecture of these distributed systems is complex that consists of a set of heterogeneous subsystems. Particularly, with the advancement of the distributed ledger technology, the nested architectures become ubiquitous in decentralized systems \cite{androulaki_hyperledger_2018} \cite{ding_dagbase_2020}. Besides, these systems collaborate with each other by complex protocols \cite{ding_derepo_2020} \cite{ding_bloccess_2020} based on peer communications and cross-layer communications  to satisfy the requirements. However, for such a complex distributed system, it is a challenge to ensure correctness and system properties.

Reasoning about distributed systems in a sound logic plays an imperative role in proving correctness and properties. As an extension of Hoare logic, separation logic (SL) introduces the separating conjunction to provide modular reasoning about systems. SL was established in papers \cite{ohearn_logic_1999} \cite{ishtiaq_bi_2001}, and firmly developed by John C. Reynolds \cite{reynolds_separation_2002}. The intention of SL is to reason about resources generally and verify the correctness of memory usage, specifically random access memory (RAM), by merging the logic model and the engineering model, which presents high value in program verification. To make SL more expressive, concurrent separation logic (CSL) was advanced by Peter W. O’Hearn \cite{ohearn_resources_2007} to reason about concurrent programs, which makes it possible to formalize thread-level or process-level parallelisms. CSL has been mechanized by recent research \cite{jung_iris_2015} \cite{krebbers_essence_2017}, \cite{bizjak_iron_2019}, which proves effectiveness and powerful expressiveness in the formalization and verification of parallel systems and distributed systems.

Nevertheless, the standard CSL does not support modularity well, which reflects in three aspects.

\begin{itemize}
\item The standard CSL only has spatial modularity, which leads to a barrier to tackle temporal problems.
\item The standard CSL focuses on the low-level formalization such as memory management, which lacks modular components such as communications.
\item The standard CSL is restricted to provide the support of nested formalization with modularity to specify and verify nested systems consisting of multiple parallel layers.
\end{itemize}

In this paper, we focus on the methodology of formalizing systems at different abstraction levels with great modularity. We propose an extended concurrent separation logic (ECSL) that enhances the modularity with the support of the temporal extension, communication extension, environment extension, and nest extension. Our principles for extending the CSL with modularity are twofold. On the one hand, we need to ensure unitarity that is to unify the specification and verification of systems at different abstraction levels. On the other hand, we need the compatibility of ECSL to permit the interpretation of CSL and typical variants.

Our logic ECSL makes the following main contributions:

\begin{enumerate}
\item ECSL has spatiotemporal modularity to facilitate both spatial and temporal reasoning with the support of the temporal extension. Furthermore, it can embed a temporal logic to formulate the temporal properties.
\item ECSL is capable of formalizing systems at different abstraction levels, especially complex systems with a nested architecture and a set of communication protocols with the support of communication extension, environment extension, and nest extension. Particularly, environment extension enables the specification to perceive the environment factors.
\item ECSL follows the unitarity and compatibility principles. It can be implemented in a specification language to formalize systems developed with an expressive programming language.
\end{enumerate}

\section{Related Work}
Since the establishment of CSL, modular reasoning of CSL has been a highlight to simplify the formalization of concurrent programs. O’Hearn addressed the importance of local reasoning in his early work \cite{ohearn_resources_2007}. Parallel mergesort was formalized in CSL with \textit{Parallel Composition Rule} to present the elegance of independent reasoning. Besides the restrictive reasoning about disjoint concurrency, the author introduced the concept of resource invariant to enable CSL to formalize inter-process interactions. A pointer-transferring buffer example was used to prove the effectiveness. Furthermore, Smallfoot \cite{berdine_smallfoot_2005} demonstrated the feasibility of mechanizing modular reasoning about concurrent programs with several detailed examples mechanized in Smallfoot. Since then, the CSL has been applied to verify a wide range of applications \cite{chen_using_2015}.

Many variants of CSL has been proposed to enhance the capability of CSL with the support of modular reasoning. In \cite{bell_concurrent_2010}, a CSL supporting shared channel endpoints was proposed for pipelined parallelization. The authors made the use of modular reasoning to formalize the asynchronous channels. Although it extended the standard CSL with communication support, it fails to solve the temporal modularity issue and high-level formalization issue.

The temporal modularity was tackled in \cite{sergey_programming_2017} by proposing DISEL, a framework for the compositional verification of distributed systems with their clients. This work extends the CSL with the capability of formalizing the distributed protocols. However, the formalization of systems consisting of multiple parallel layers is still a challenge.

Recently, Iris Project, a higher-order CSL, was developed \cite{jung_iris_2015,jung_higher-order_2016,krebbers_essence_2017}, which highly extended CSL and enhanced modularity both in aspects of temporal formalization and high-level formalization. Particularly, Actris \cite{hinrichsen_actris_2019} and Aneris \cite{krogh-jespersen_aneris_2020} were proposed for reasoning about message passing and node-local resources in distributed systems. However, Aneris still struggles with the formalization of nested parallel architecture.

\section{Concurrent Separation Logic}
In this section, we introduce some critical concepts of CSL as the preliminary of our logic.

As an extension of SL, CSL acts as a concurrent program logic in proving the correctness properties of concurrent programs. It still uses a Hoare triple style as the form for proving specifications: $\{ P \} ~\alpha~ \{ P' \}$, $P$ and $P'$ denote pre-condition and post-condition respectively while $\alpha$ denotes the action that changes the state of the program.

We consider a general CSL introduced in \cite{vafeiadis_concurrent_2011}. Program state is defined as a tuple $\langle S,H \rangle$, where $S$ denotes the stack and $H$ denotes the heap. A specification language is structured in Definition~\ref{def:spec_lang_csl}.

\begin{definition}[Syntax of Specification Language of CSL]
\label{def:spec_lang_csl}
Let $\bar{E}$ and $\bar{B}$ denote arithmetic and Boolean expressions respectively. The structure of the assertion $P$ and $P'$ are defined as follows:

$\check{P}, \check{P}' ::= \textbf{emp} ~|~ \bar{B} ~|~ \bar{E}_1 \mapsto \bar{E}_2 ~|~ \check{P} \land \check{P}' ~|~ \check{P} \lor \check{P}' ~|~ \neg \check{P} ~|~ \check{P} \implies \check{P}' ~|~ \check{P} * \check{P}' ~|~ \check{P} \mathrel{-\mkern-6mu*} \check{P}'$
\end{definition}

Separating conjunction $*$ and separation implication $\mathrel{-\mkern-6mu*}$ are two critical operators with special semantics. We define a modeling relation $(s, h) \models \check{P}, s \in S, h \in H$ meaning that the program state $(s, h)$ satisfies the assertion $\check{P}$.

\begin{definition}[Semantics of Assertions]
\label{def:sem_CSL}
The semantics of assertions in CSL is given as follow:

$(s,h) \models \textit{emp} \iff \textit{Dom}(h) = \emptyset$

$(s,h) \models \bar{E}_1 \mapsto \bar{E}_2 \iff \textit{Dom}(h) = \llbracket \bar{E}_1 \rrbracket_s \land h(\llbracket \bar{E}_1 \rrbracket_s) = \llbracket \bar{E}_2 \rrbracket_s $

$(s,h) \models \check{P} * \check{P}' \iff \exists h_1,h_2: h=h_1 \uplus h_2 \land (s,h_1) \models \check{P} \land (s,h_2) \models \check{P}'$

$(s,h) \models \check{P} \mathrel{-\mkern-6mu*} \check{P}' \iff \forall h_1: (\widetilde{h \uplus h_1}) \land (s,h_1) \models \check{P} \implies (s,h \uplus h_1) \models \check{P}'$

$\llbracket \bar{E} \rrbracket_s \triangleq s(\bar{E})$

$\widetilde{h} \triangleq h \text{ is defined}$
\end{definition}

In CSL, an important proof rule is \textit{Parallel Composition Rule} given as follow:

\begin{align*}
\infer{\{ \circledast_{i=0}^n P_i \} ~\alpha_0 \parallel ... \parallel \alpha_n ~ \{ \circledast_{i=0}^n P_i' \} }
{\{ P_0 \}~\alpha_0~\{ P_0' \}~...~\{ P_n \}~\alpha_n~\{ P_n' \} } \\
\text{(Parallel Composition Rule)}
\end{align*}

$\circledast_i^n$ denotes consecutive separating conjunction from index $i$ to $n$.

This rule is the key to formalize the disjoint concurrency with the support of completely local reasoning about processes in a parallel program. Furthermore, CSL gives \textit{Critical Region Rule} to reason about the inter-process interaction.

\begin{align*}
\infer{\{ P \} \text{ with } r \text{ when } B \text{ do } \alpha~\{ P' \} }
{\{ (P * \textit{RI}_r) \land B \}~\alpha~\{ P' * \textit{RI}_r \} } \\
\text{(Critical Region Rule)}
\end{align*}

Here, $\textit{RI}_r$ denotes the resource invariant and $B$ is the guard. Resource $r$ provides a mutual exclusion for different interactions with critical regions in a program.

We can obtain that CSL supports modular reasoning by introducing separation operators. However, a general CSL only focuses on the spatial formalization with low-level modular reasoning about the memory.

\section{Extended Concurrent Separation Logic}

\subsection{Program Model}
To illustrate our logic, we firstly define a program model in Definition \ref{def:program}.

\begin{definition}
\label{def:program}
A program over set $V$ of typed variables is defined as

$\mathfrak{P} \triangleq (L, A, \mathcal{E}, \hookrightarrow, L_0, g_0)$,

where $L$ is a set of code locations and $A$ is a set of actions. $\mathcal{E}$ denotes the effect function $A \times \llbracket V \rrbracket \mapsto \llbracket V \rrbracket$. The notation ${\hookrightarrow} \subseteq L \times \| V \| \times A \times L$ represents the conditional transition relation. $L_0 \subseteq L$ and $g_0 \in \| V \|$ denotes a set of initial locations and the initial condition respectively. $\llbracket V \rrbracket$ denotes the set of variable evaluations that include memory locations $\mathcal{L}$. $\| V \|$ denotes the set of Boolean conditions over $V$.
\end{definition}

For convenience, we use the notation $l \xhookrightarrow{g:\alpha} l'$ as shorthand for $(l,g,\alpha,l') \in {\hookrightarrow}$ where $l \in L$ and $\alpha \in A$, meaning that the program $\mathfrak{P}$ goes from location $l$ to $l'$ when the current variable evaluation $\eta \models g$. Therefore, we can specify $l \xhookrightarrow{g:\alpha} l'$ in CSL as $\{ g \} \alpha \{ g' \}$ where $\mathcal{E}(\alpha,\eta) \models g'$.

We call a program consisting of a set of programs as a system that is denoted as $\mathfrak{W}$.

\subsection{Temporal Extension}
We introduce temporal representation and reasoning to extend the CSL into two-dimension reasoning with temporal memory as the composition of program states. To illustrate the temporal memory, we firstly define \textit{action occurrence} in Definition~\ref{def:action_occurence}.

\begin{definition}[Action Occurrence]
\label{def:action_occurence}
An action occurrence is a partial function $A \rightharpoonup \acute{A}$, where $\acute{A}$ denotes the set of occurred actions. $a \in \acute{A}$ is a tuple $\langle I, O, A, I_p \rangle$, where $I \subseteq \mathbb{N}^{+}$ is the set of indices of occurred actions, and $O \subseteq \mathbb{N}$ denotes the set of indices of action executors that are programs executing action $a$. $I_p \subseteq \mathbb{N}$ denotes the set of indices of occurred actions happened before action $a$. If $a$ is the first occurred action, then $a.I_p = 0$.
\end{definition}

Now, we formally define the structure of the program state in ECSL.

\begin{definition}[Program State]
\label{def:program_state}
A program state $\omega \in \Omega$ is a tuple $\langle S, H, T \rangle$, where $S$ denotes the stack, $H$ denotes the heap, and $T$ denotes the temporal memory. Formally, we define $S$, $H$, and $T$ as follows:

\begin{align*}
S \triangleq V \mapsto \llbracket V \rrbracket \\
H \triangleq \mathcal{L} \rightharpoonup_\text{fin} \llbracket V \rrbracket \\
T \triangleq I \mapsto \acute{A} \\
\Omega \triangleq S \times H \times T
\end{align*}

\end{definition}

We discuss further the temporal memory by defining path and trace. Firstly, we define a relationship between actions.

\begin{definition}[Action Relation]
The action ordered relation $\triangleleft$ is a partially ordered relation, which is defined as:

$(a,a') \in \triangleleft \iff a=\textit{Pre}(a')$,

where $a,a' \in \acute{A}$ and $\textit{Pre}(a)$ denotes the predecessor action set of $a$.
\end{definition}

We use the notation $a \triangleleft a'$ as shorthand for $(a,a') \in \triangleleft$. Intuitively, $a \triangleleft a'$ means action $a$ happens before action $a'$.

\begin{definition}[Action Path]
\label{def:action_path}
A finite action path $\hat{\varrho}$ is a finite action sequence $a_1 a_2 ... a_n$ such that $\forall i \in [1,n): (a_i,a_{i+1}) \in \triangleleft$, where $n \geq 1$. An infinite action path $\varrho$ of $\mathfrak{P}$ is an infinite action sequence $a_1 a_2 a_3 ...$ such that $\forall i \in [1,+\infty):(a_i,a_{i+1}) \in \triangleleft$.
\end{definition}

\begin{definition}[Maximal and Initial]
\label{def:max_init}
An action path is maximal if and only if it is finite and terminable or infinite. An action path is initial if and only if $l_0 \xhookrightarrow{g_0:\alpha} l$, where $l_0 \in L_0$ is the initial location, $g_0$ is the initial condition, and $\alpha = a_1.A$ is the action of the first occurred action.

\end{definition}

Now, we can define temporal memory in Definition~\ref{def:temporal_memory} with Definition~\ref{def:action_path} and Definition~\ref{def:max_init}.

\begin{definition}[Temporal Memory]
\label{def:temporal_memory}
Temporal memory $t \in T$ is an initial and maximal action path.
\end{definition}

We consider a program $\mathfrak{P}$ in $\mathfrak{W}$. For $a \in t$, if $a.O$ is mapping to $\mathfrak{P}$, we call $a$ is a block in the temporal memory associated with $\mathfrak{P}$. All blocks associated with $\mathfrak{P}$ compose a subset of the temporal memory $t_n$ that is the native temporal memory of $\mathfrak{P}$. We have $t = t_f \uplus t_n$ where $t_f$ is the foreign temporal memory of $\mathfrak{P}$. To connect the semantics with the temporal properties that are verified in ECSL, we introduce the action trace defined in Definition~\ref{def:action_trace}.

\begin{definition}[Action Trace]
\label{def:action_trace}
The action trace of the finite action path $\hat{\varrho}=a_1 a_2 ... a_n$ is defined as $\mathcal{T}(\hat{\varrho})=\mathcal{P}(a_1)\mathcal{P}(a_2)...\mathcal{P}(a_n)$. The action trace $\mathcal{T}(\varrho)$ of the infinite action path $\varrho$ is defined in the same way.

$\mathcal{P}$ is a function that relates a set of propositions to occurred action $a$, which is defined as $\mathcal{P} \triangleq a \mapsto 2^\textit{Prop}$, where $\textit{Prop}$ is a set of propositions.

\end{definition}

Additionally, we use $\mathcal{T}(\mathfrak{P})$ to denote the set of all possible traces of program $\mathfrak{P}$. The set of all possible traces of the system is denoted as $\mathcal{T}(\mathfrak{W})$.

\begin{theorem}
\label{the:trace_satisfication}
Let \textit{Prop} be a set of propositions over a finite action trace $\mathcal{T}(\hat{\varrho})$ and $\Phi$ be a propositional logic formula over \textit{Prop}, then

$\infer{\mathcal{T}(\hat{\varrho}) \models \Phi}
{\forall a' \in \hat{\varrho}: (\forall a \in \hat{\varrho}: a \triangleleft a' \land \mathcal{P}(a) \models \Phi \implies \mathcal{P}(a') \models \Phi)}$
\end{theorem}

\begin{proof}
Let us consider a finite action trace $\mathcal{T}(\hat{\varrho})$ of the finite action path $\hat{\varrho}=a_1a_2...a_n$. We take an action $a_i \in \hat{\varrho}$ to construct a fragment $\hat{\varrho}'=a_1...a_i$ of the path such that $\forall a \in \hat{\varrho}': a \triangleleft a' \land \mathcal{P}(a) \models \Phi$. If it implies that $\mathcal{P}(a_i) \models \Phi$, we have $\forall a \in \hat{\varrho}': \mathcal{P}(a) \models \Phi \iff \mathcal{T}(\hat{\varrho}') \models \Phi$.

Therefore, we can take all $a' \in \hat{\varrho}$ to construct all possible fragments of the action path. If $\forall \hat{\varrho}' \subseteq \hat{\varrho}:(\forall a \in \hat{\varrho}': a \triangleleft a' \land \mathcal{P}(a) \models \Phi \implies \mathcal{P}(a') \models \Phi)$, then $\forall a \in \hat{\varrho}: P(a) \models \Phi \iff \mathcal{T}(\hat{\varrho}) \models \Phi$.
\end{proof}

Intuitively, Theorem \ref{the:trace_satisfication} shows that for any occurred action $a'$ in $\mathcal{T}$, if all actions that happen before $a'$ satisfy $\Phi$, which implies that $a'$ also satisfies $\Phi$, then we have trace $\mathcal{T}$ satisfies $\Phi$.

\subsection{Communication Extension}
In ECSL, we consider the formalization of the communication as the elementary component to facilitate the high-level formalization of a complex system where communications among programs are indispensable.

Channel is the abstraction of the transmission medium to convey information signals. Communication is the action of passing and receiving messages through the channel. In ECSL, we reason about the channels and communications as the basis.

\begin{definition}[Channel]
A channel $c \in C$ is a buffer with a capacity $\textit{Cap}(c) \in \mathbb{N} \cup \{ \infty \}$, a domain $\textit{Dom}(c)$.
\end{definition}

Let $c!m$ denote sending signal $m$ via channel $c$ and $c?v$ denote receiving a signal from channel $c$ and assign the signal to variable $v$.

\begin{definition}[Communication]
A communication $\pi \in \Pi$ is an action where $\Pi = \{ c!m, c?v \}$, $c \in C, m \in \textit{Dom}(c), v \in V \text{ with } \textit{Dom}(v) \supseteq \textit{Dom}(c)$.
\end{definition}

\begin{definition}[Complete]
A communication $c!m$ is complete if and only if there exists a communication $c?v$ satisfying that $(c!m, c?v) \in \triangleleft$. A communication $c?v$ is complete if and only if there exists a communication $c!m$ satisfying that $(c!m, c?v) \in \triangleleft$.
\end{definition}

We can obtain that it is reasonable to consider channel $c$ as a buffer. In this manner, a communication $c!m$ produces signal $m$ into the buffer whereas a communication $c?v$ consumes a signal from the buffer while assigning it to variable $v$.

With the definition of channels and communications, we extend our program model into the channel program model.

\begin{definition}[Channel Program]
A channel program over $(V, C)$ is defined in the same manner as Definition~\ref{def:program}:

$\mathfrak{C} \triangleq (L, A, \mathcal{E}, \hookrightarrow, L_0, g_0)$.

The only difference is that $\hookrightarrow \subseteq L \times \| V \| \times (A \cup \Pi) \times L$ where $V$ is a set of typed variables and $C$ is a set of channels.
\end{definition}

Hence, conditional transitions are extended with the communication action set $\Pi$, which yields conditional transitions $l \xhookrightarrow{g:c!m} l'$ and $l \xhookrightarrow{g:c?v} l'$ respectively.

\begin{theorem}
\label{thm:comm_spec}
Let $c \in C$ be a finite asynchronous channel. If a complete communication $c!m$ and $c?v$ is correct, then the specifications below are satisfied.

\begin{enumerate}
    \item $\{ m \mapsto - \}~c!m~\{ \top \}$
    \item $\{ v \mapsto - \}~c?v~\{ \llbracket v \rrbracket = \llbracket m \rrbracket \}$
\end{enumerate}
\end{theorem}

\begin{proof}
For the finite asynchronous channel $c$, we consider that $c$ is a finite queue with $\textit{Cap}(c) \in \mathbb{N}^{+}$. Recall the basic data type operation of a queue such as $\textit{Enqueue}$ (appending the element at the rear of the queue), $\textit{Dequeue}$ (returning and removing the element at the head of the queue), $\textit{Peek}$ (returning the element at the head of the queue without removing the element from the queue).

We specify $c!m$ and $c?v$ respectively as follows.

$c!m \triangleq \text{with } c \text{ when } \neg \textit{full}$ do $\textit{Enqueue}(c, m)$

$c?v \triangleq \text{with } c \text{ when } \neg \textit{empty}$ do $v \gets \textit{Dequeue}(c)$

where $\textit{full} \triangleq \textit{Len}(c) = \textit{Cap}(c)$, $\textit{empty} \triangleq \textit{Len}(c) = 0$.

We introduce the resource invariant $RI_{c}$ as follows:

$\textit{RI}_c \triangleq \neg \textit{full} \lor (\neg \textit{empty} \land \textit{Peek}(c) = \llbracket m \rrbracket)$

We have channel proof rules derived from \textit{Critical Region Rule}.

\begin{align*}
\infer{\{ P \} ~c!m~ \{ P' \} }
{\{ (P * \textit{RI}_c) \land \neg \textit{full} \}~\textit{Enqueue}(c,m)~\{ P' * \textit{RI}_c \} } \\
\infer{\{ P \} ~c?v~ \{ P' \} }
{\{ (P * \textit{RI}_c) \land \neg \textit{empty} \}~v \gets \textit{Dequeue}(c,m)~\{ P' * \textit{RI}_c \} }
\label{eq:cpr}
\end{align*}

Now, we prove the first specification $\{ m \mapsto - \}~c!m~\{ \top \}$.

\begin{align*}
    \{ (\textit{RI}_c * m \mapsto -) \land \neg \textit{full} \} \\
    \{ \neg \textit{full} * m \mapsto - \} \\
\textit{Enqueue}(c, \llbracket m \rrbracket) \\
    \{ \neg \textit{empty} \land \textit{Peek}(c) = \llbracket m \rrbracket \} \\
    \{ \textit{RI}_c \} \\
    \{ \textit{RI}_c \land \top \}
\end{align*}

The proof of $\{ v \mapsto - \}~c?v~\{ \llbracket v \rrbracket = \llbracket m \rrbracket \}$ is given below.

\begin{align*}
    \{ (\textit{RI}_c * v \mapsto -) \land \neg \textit{empty} \} \\
    \{ (\neg \textit{empty} \land \textit{Peek}(c) = \llbracket m \rrbracket) * v \mapsto - \} \\
    v \gets \textit{Dequeue}(c) \\
    \{ \llbracket v \rrbracket = \llbracket m \rrbracket \land \neg \textit{full} \} \\
    \{ \neg \textit{full} * \llbracket v \rrbracket = \llbracket m \rrbracket \} \\
    \{ \textit{RI}_c * \llbracket v \rrbracket = \llbracket m \rrbracket \}
\end{align*}

This completes our proof of the specifications given the complete communication is correct.
\end{proof}

\begin{corollary}
\label{cor:com_spec}
Let $c \in C$ be a finite asynchronous channel. If a complete communication $c!m$ and $c?v$ is correct, then $\{ m \mapsto - \}~c!m \parallel c?v~\{ \llbracket v \rrbracket = \llbracket m\rrbracket \}$ is satisfied.
\end{corollary}

\begin{proof}
With \textit{Parallel Composition Rule}, we only need to prove separate programs locally in the same way of proof of Theorem~\ref{thm:comm_spec} and combine these local proofs with \textit{Critical Region Rule} and \textit{Parallel Composition Rule}. The proof outline is given in Table~\ref{tab:com_proof_spec2}.

\begin{table}
  \caption{Proof outline of $\{ m \mapsto - \}~c!m \parallel c?v~\{ \llbracket v \rrbracket = \llbracket m \rrbracket \}$.}
  \label{tab:com_proof_spec2}
  \centering
  \begin{tabular}{ccc}
  \multicolumn{3}{c}{$\{ m \mapsto - \}$} \\
  \multicolumn{3}{c}{$\{ m \mapsto - * \top \}$} \\
  $\{ m \mapsto - \}$ &  & $\{ \top \}$ \\
  $c!m$ & $\parallel$ & $c?v$ \\
  $\{ \top \}$ & & $\{ \llbracket v \rrbracket = \llbracket m \rrbracket \}$ \\
  \multicolumn{3}{c}{$\{ \top * \llbracket v \rrbracket = \llbracket m \rrbracket \}$} \\
  \multicolumn{3}{c}{$\{ \llbracket v \rrbracket = \llbracket m \rrbracket \}$}
  \end{tabular}
\end{table}
\end{proof}

\subsection{Environment Extension}
We extend the CSL with the representation of the environment to reason about the environment factors of a program, especially in a parallel system. The environment factor can be factorized into foreign factor and native factor.

\begin{definition}[Judgement Form]
\label{def:judgement_form}
We define the form of a judgement of ECSL as:

$J \vdash \{ \Gamma, \gamma \land P \}~\alpha~\{ \Gamma, \gamma' \land P' \}$,

where $J$ denotes the judgement \cite{reddy_syntactic_2012}. $\Gamma$ specifies the foreign conditions while $\gamma$ and $\gamma'$ specify the native pre-conditions and post-conditions. $P$ and $P'$ are assertions. $\alpha \in A$ is the action to change the state of programs. The syntax of $\Gamma$ and $\gamma$ is defined in a manner of temporal logic.
\end{definition}

Intuitively, the foreign environment is a set of conditions from other processes while the native environment is a set of conditions from the local process.

There is a special kind of assertion called \textit{pure} assertion.

\begin{definition}[Pure Assertion]
\label{def:pure}
In program $\mathfrak{P}$, an assertion $P$ is pure if and only if

$(s, h, t) \models P \implies \forall t' \in T: (s, h, t') \models P$.
\end{definition}

In intuitive terms, assertion $P$ is pure if and only if the validity of $P$ is independent with the environment factors. If both pre-condition and post-condition only contain pure assertions, the specification is reduced to the CSL.

We introduce new rules for environment extension. For brevity, we use $\Upsilon$ to denote the conjunction of $\gamma$ and $P$. The big star notation $\circledast_i^n$ denotes consecutive separating conjunction from index $i$ to $n$.

\begin{align*}
\infer{\{ \Gamma, \Upsilon \} \text{ when } \Gamma \text{ do } \alpha~\{ \Gamma, \Upsilon' \} }
{\{ \Gamma, \Upsilon \}~\alpha~\{ \Gamma, \Upsilon' \} } \\
\text{(Foreign Environment Rule)}
\end{align*}
\begin{align*}
\infer{\{ \circledast_{i=0}^n \Upsilon_i \}~\alpha_0 \parallel ... \parallel \alpha_n~\{ \circledast_{i=0}^n \Upsilon_i' \}}
{\{ \Gamma_0, \Upsilon_0 \}~\alpha_0~\{ \Gamma_0, \Upsilon_0' \}~...~\{ \Gamma_n, \Upsilon_n \}~\alpha_n~\{ \Gamma_n, \Upsilon_n' \}} \\
\text{(Environment Composition Rule)}
\end{align*}

In \textit{Environment Composition Rule}, the inference eliminates the foreign environment naturally if we regard a parallel system as the highest level of specification, which means that all separated programs must mutually satisfy the foreign environments of the other.

\begin{example}
\label{eg:env}
We consider a system $\mathfrak{W}$ containing two parallel channel programs $\mathfrak{C}_0$ and $\mathfrak{C}_1$ specified in Table~\ref{tab:pcs_spec} where the left channel program is $\mathfrak{C}_0$ and the right one is $\mathfrak{C}_1$. The specification is well formulated with the support of the communication encapsulation and Theorem~\ref{thm:comm_spec}.

\begin{table*}
  \caption{Specification of $\mathfrak{W}$ in Example~\ref{eg:env}.}
  \label{tab:pcs_spec}
  \centering
  \begin{tabular}{ccc}
    & $\{ \top, m_0, m_1 \mapsto -, - * v_0, v_1 \mapsto -, - \}$ & \\
    $c!m_0$ & $\parallel$ & $c?v_0$ \\
    $c!m_1$ & & $c?v_1$ \\
   \multicolumn{3}{c}{$\{ \top, (c!m_0 \triangleleft c!m_1 \land \top) * (c?v_0 \triangleleft c?v_1 \land \llbracket v_0 \rrbracket = \llbracket m_0 \rrbracket \land \llbracket v_1 \rrbracket = \llbracket m_1 \rrbracket) \}$}
  \end{tabular}
\end{table*}

The sending channel program $\mathfrak{C}_0$ sends signals $m_0$, $m_1$. The receiving channel program $\mathfrak{C}_1$ receives signal $m_0$ and assigns it to local variable $v_0$ after $\mathfrak{C}_0$ completes sending action $c!m_0$ and receives signal $m_1$ in the same way.

We can verify $\mathfrak{W}$ by verifying $\mathfrak{C}_0$ and $\mathfrak{C}_1$ separately according to \textit{Environment Composition Rule}. For each channel program, we formalize it from both foreign perspective and native perspective with the support of environment factors.

For $\mathfrak{C}_0$ that is sending signals through channel $c$, we can specify and verify it as follows with omission of detailed pointer management:

\begin{align*}
    \{ \top, \top \land m_0, m_1 \mapsto -, - \} \\
    c!m_0 \\
    \{ \top, c!m_0 \land m_1 \mapsto - \} \\
    c!m_1 \\
    \{ \top, c!m_0 \triangleleft c!m_1 \land \top \}
\end{align*}

It is noteworthy that with Definition~\ref{def:pure}, we can reduce the formalization of $\mathfrak{C}_0$ to a CSL form by omitting the environment factor because the assertion is pure.

We use the action path as shorthand to notate the atomic proposition in a temporal logic. For $\mathfrak{C}_1$ that is receiving signals from channel $c$ and assigning signals to different variables, we can formalize it as follows:

\begin{align*}
    \{ \top, v_0, v_1 \mapsto -, - \} \\
    \{ c!m_0, v_0, v_1 \mapsto -, - \} \\
    c?v_0 \\
    \{ c!m_0, c?v_0 \land \llbracket v_0 \rrbracket = \llbracket m_0 \rrbracket \} \\
    \{ c!m_0 \triangleleft c!m_1, c?v_0 \land \llbracket v_0 \rrbracket = \llbracket m_0 \rrbracket \} \\
    c?v_1 \\
    \{ c!m_0 \triangleleft c!m_1, c?v_0 \triangleleft c?v_1 \land \llbracket v_0 \rrbracket = \llbracket m_0 \rrbracket \land \llbracket v_1 \rrbracket = \llbracket m_1 \rrbracket \}
\end{align*}

Now, we give the outline of proof in Table~\ref{tab:env_proof} with \textit{Environment Composition Rule}.

\begin{table*}
  \caption{Proof outline of the specification of $\mathfrak{W}$ in Example~\ref{eg:env}.}
  \label{tab:env_proof}
  \centering
  \begin{tabular}{lcl}
    & $\{ \top, m_0, m_1 \mapsto -, - * v_0, v_1 \mapsto -, - \}$ & \\
    $\{ \top, \top \land m_0, m_1 \mapsto -, - \}$ & & $\{ c!m_0, v_0, v_1 \mapsto -, - \}$ \\
    $c!m_0$ & $\parallel$ & $c?v_0$ \\
    $\{ \top, c!m_0 \land m_1 \mapsto - \}$ & & $\{ c!m_0, c?v_0 \land \llbracket v_0 \rrbracket = \llbracket m_0 \rrbracket \}$ \\
    $\{ \top, c!m_0 \land m_1 \mapsto - \}$ & & $\{ c!m_0 \triangleleft c!m_1, c?v_0 \land \llbracket v_0 \rrbracket = \llbracket m_0 \rrbracket \}$ \\
    $c!m_1$ & $\parallel$ & $c?v_1$ \\
    $\{ \top, c!m_0 \triangleleft c!m_1 \land \top \}$ & & $\{ c!m_0 \triangleleft c!m_1, c?v_0 \triangleleft c?v_1 \land \llbracket v_0 \rrbracket = \llbracket m_0 \rrbracket \land \llbracket v_1 \rrbracket = \llbracket m_1 \rrbracket \}$ \\
    \multicolumn{3}{c}{$\{ \top, (c!m_0 \triangleleft c!m_1 \land \top) * (c?v_0 \triangleleft c?v_1 \land \llbracket v_0 \rrbracket = \llbracket m_0 \rrbracket \land \llbracket v_1 \rrbracket = \llbracket m_1 \rrbracket) \}$}
  \end{tabular}
\end{table*}

\end{example}

From Example~\ref{eg:env}, we can obtain that each program in a system can be proved locally and then combined together with \textit{Environment Composition Rule} as long as their native environments could mutually satisfy the foreign environments. In other words, the system can be proved correct if and all foreign environments of programs in the system can be eliminated.

By extending the environment representation, we equip the CSL with the capability to reason about the assertions together with environment factors including foreign environment and native environment, which enhances the modularity. Furthermore, we can formulate temporal properties in systems explicitly.

\subsection{Nest Extension}
To formalize complex systems with a nested architecture, we introduce the nest extension to enhance the capability of formalizing these systems from a low abstraction level to a high abstraction level.

Let $\mathfrak{N}$ denote a set of systems at the same level. We have $N_0,.., N_n \in \mathfrak{N}$. $@$ is the ownership relation between action and system. $(\alpha, N) \in @$ denotes action $\alpha$ happens at system $N$, meaning that system $N$ has the ownership of action $\alpha$. We use the notation $\alpha @ N$ as shorthand for $(\alpha,N) \in @$. $\parallel_i$ is the notation for nest parallel where $i$ denotes the level ID to distinguish with program parallel $\parallel$.

We introduce new rules for nest extension.

\begin{align*}
\infer{\{ \hat{\Gamma}, \circledast_{i=0}^n \Upsilon_i \}~\alpha_0@N \parallel ... \parallel \alpha_n@N~\{ \hat{\Gamma}, \circledast_{i=0}^n \Upsilon_i' \}}
{\{ \Gamma_0, \Upsilon_0 \}~\alpha_0@N~\{ \Gamma_0, \Upsilon_0' \}~...~\{ \Gamma_n, \Upsilon_n \}~\alpha_n@N~\{ \Gamma_n, \Upsilon_n' \}} \\
\text{(Nest Environment Composition Rule)}
\end{align*}

\begin{align*}
\infer{ \{ \circledast_{i=0}^n \Upsilon_i \}~\alpha_0@N_0 \parallel_i ... \parallel_i \alpha_n@N_n~\{ \circledast_{i=0}^n \Upsilon_i' \} }
{ \{ \Gamma_0, \Upsilon_0 \}~\alpha_0@N_0~\{ \Gamma_0, \Upsilon_0' \}...\{ \Gamma_n, \Upsilon_n \}~\alpha_n@N_n~\{ \Gamma_n, \Upsilon_n' \}} \\
\text{(Nest Composition Rule)}
\end{align*}

The most distinguishing feature of \textit{Nest Environment Composition Rule} is the immutability of the foreign environment. For a subsystem on $N$, the foreign environment includes the temporal properties of other subsystems on $N$ and the temporal properties of other systems at the same level. While specifying a parallel subsystem on system $N$, the temporal properties of other systems can be used as environment factors. For instance, to do action $c!m_1$, the sending program on $N$ needs to satisfy that another program on $N$ has received a signal $m_0$ from $N'$ through channel $c$ and assigned $m_0$ to variable $v$, which can be specified as $\{ c!m_0@N' \triangleleft c?v@N, m_1 \mapsto - \}~c!m_1~\{ \top, c!m_1 \}$. In this manner, we distinguish communications at different levels.

Furthermore, our \textit{Nest Composition Rule} makes it possible to construct a nested system having multiple abstraction levels.

\begin{example}
\label{eg:nest}
We consider a small network with two systems $N_0$ and $N_1$ communicating with each other. Each system has two parallel channel programs with one for sending signals and another one for receiving signals.

Assume a simple communication protocol that $N_0$ sends a message $m_0$ to $N_1$ via channel $c_0$. When $N_1$ receives message $m_0$ from $N_0$, $N_1$ sends message $n_0$ to $N_0$ through channel $c_1$. When $N_0$ receives message $n_0$ from channel $c_1$, $N_0$ sends message $m_1$ back to $N_1$ through channel $x_1$. When $N_0$ confirms that it has received $n_1$ and finished the assignment, the sending channel program of $N_0$ will terminate in normal.

We annotate interactions in the network in Table~\ref{tab:nest_annotate}.

\begin{table*}
\caption{Annotation of the network in Example~\ref{eg:nest}.}
\label{tab:nest_annotate}
\centering
\begin{tabular}{ccccccc}
    $c_0!m_0 @ N_0$ & & $c_1?x_0 @ N_0$ & & $c_1!n_0 @ N_1$ & & $c_0?y_0 @ N_1$ \\
    & $\parallel$ & & $\parallel_0$ & & $\parallel$ & \\
    $c_0!m_1 @ N_0$ & & $c_1?x_1 @ N_0$ & & $c_1!n_1 @ N_1$ & & $c_0?y_1 @ N_1$
  \end{tabular}
\end{table*}

With the support of the nest parallel extension, we can specify systems separately and combine local reasoning to complete the specification and proof. We specify system $N_0$ as follows:
\begin{align*}
\{ \top, m_0, m_1 \mapsto -, - * x_0, x_1 \mapsto -, - \} \\
c_0!m_0@N_0;c_0!m_1@N_0 \parallel c_1?x_0@N_0;c_1?x_1@N_0 \\
\{ c_1!n_0@N_1 \triangleleft c_1!n_1@N_1, (c_0!m_0 \triangleleft c_0!m_1) * \\
(c_1?x_0 \triangleleft c_1?x_1 \land \llbracket x_0 \rrbracket = \llbracket n_0 \rrbracket \land \llbracket x_1 \rrbracket = \llbracket n_1 \rrbracket) \}
\end{align*}

We give the proof outline for the specification of $N_0$ in Table~\ref{tab:n_0_proof}. For brevity, we omit the pointer management and merge some pre-conditions and post-conditions that are trivial to be proved with the CSL inference rules.

\begin{table*}
\caption{Proof outline of the specification of $N_0$ in Example~\ref{eg:nest}.}
\label{tab:n_0_proof}
\centering
  \begin{tabular}{ccc}
  & $\{ \top, m_0, m_1 \mapsto -, - * $ & \\
  & $x_0, x_1 \mapsto -, - \}$ & \\
  $\{ \top, m_0, m_1 \mapsto -, - \}$ & & $\{ c_0!m_0@N_0 \triangleleft c_1!n_0@N_1, x_0, x_1 \mapsto -, - \}$ \\
  $c_0!m_0@N_0$ & $\parallel$ & $c_1?x_0@N_0$ \\
  $\{ c_1?x_0@N_0, c_0!m_0@N_0 \land m_1 \mapsto - \}$ & & $\{ c_0!m_1@N_0 \triangleleft c_1!n_1@N_1, \llbracket x_0 \rrbracket = \llbracket n_0 \rrbracket \land x_1 \mapsto - \}$ \\
  $c_0!m_1@N_0$ & $\parallel$ & $c_1?x_1@N_0$ \\
  $\{ c_1?x_1@N_0, c_0!m_0@N_0 \triangleleft c_0!m_1@N_0 \}$ & & $\{ c_1!n_0@N_1 \triangleleft c_1!n_1@N_1, c_1?x_0@N_1 \triangleleft c_1?x_1@N_1 \land $ \\
  & & $\llbracket x_0 \rrbracket = \llbracket n_0 \rrbracket \land \llbracket x_1 \rrbracket = \llbracket n_1 \rrbracket \}$ \\
  \multicolumn{3}{c}{$\{ c_1!n_0@N_1 \triangleleft c_1!n_1@N_1, (c_0!m_0@N_0 \triangleleft c_0!m_1@N_0) * (c_1?x_0@N_0 \triangleleft c_1?x_1@N_0 \land \llbracket x_0 \rrbracket = \llbracket n_0 \rrbracket \land \llbracket x_1 \rrbracket = \llbracket n_1 \rrbracket) \}$}
  \end{tabular}
\end{table*}

We also give the specification of system $N_1$ as follows:
\begin{align*}
\{ \top, y_0, y_1 \mapsto -,- * n_0, n_1 \mapsto -, - \} \\
c_0?y_0@N_1;c_0?y_1@N_1 \parallel c_1!n_0@N_1;c_1!n_1@N_1 \\
\{ c_0!m_0@N_0 \triangleleft c_0!m_1@N_0, (c_1!n_0 \triangleleft c_1!n_1) * \\
(c_0?y_0 \triangleleft c_0?y_1 \land \llbracket y_0 \rrbracket = \llbracket m_0 \rrbracket \land \llbracket y_1 \rrbracket = \llbracket m_1 \rrbracket) \} \}
\end{align*}

The proof outline is similar to system $N_0$ in Table~\ref{tab:n_1_proof}.
\begin{table*}
\caption{Proof outline of the specification of $N_1$ in Example~\ref{eg:nest}.}
\label{tab:n_1_proof}
\centering
  \begin{tabular}{ccc}
  & $\{ \top, y_0, y_1 \mapsto -,- * $ & \\
  & $n_0, n_1 \mapsto -, - \}$ & \\
  $\{ c_0!m_0@N_0, y_0, y_1 \mapsto -, - \}$ & & $\{ c_0?y_0@N_1, n_0, n_1 \mapsto -, - \}$ \\
  $c_0?y_0@N_1$ & $\parallel$ & $c_1!n_0@N_1$ \\
  $\{ c_1!n_0@N_1 \triangleleft c_0!m_1@N_0, \llbracket y_0 \rrbracket = \llbracket m_0 \rrbracket \land y_1 \mapsto - \}$ & & $\{ c_0?y_1@N_1, c_1!n_0@N_1 \land n_1 \mapsto - \}$ \\
  $c_0?y_1@N_1$ & $\parallel$ & $c_1!n_1@N_1$ \\
  $\{ c_0!m_0@N_0 \triangleleft c_0!m_1@N_0, c_0?y_0 \triangleleft c_0?y_1 \land $ & & $\{ \top, c_1!n_0@N_1 \triangleleft c_1!n_1@N_1 \}$ \\
  $\llbracket y_0 \rrbracket = \llbracket m_0 \rrbracket \land \llbracket y_1 \rrbracket = \llbracket m_1 \rrbracket \}$ & & \\
  \multicolumn{3}{c}{$\{ c_0!m_0@N_0 \triangleleft c_0!m_1@N_0, (c_1!n_0@N_1 \triangleleft c_1!n_1@N_1) * (c_0?y_0@N_1 \triangleleft c_0?y_1@N_1 \land \llbracket y_0 \rrbracket = \llbracket m_0 \rrbracket \land \llbracket y_1 \rrbracket = \llbracket m_1 \rrbracket) \}$}
  \end{tabular}
\end{table*}

Let $P(N_i)'$ denote the post-condition of system $N_i$. In this network, we have $P(N_0)'$ and $P(N_1)'$ as follows:

$P(N_0)' = \{ c_1!n_0@N_1 \triangleleft c_1!n_1@N_1, (c_0!m_0@N_0 \triangleleft c_0!m_1@N_0) * (c_1?x_0@N_0 \triangleleft c_1?x_1@N_0 \land \llbracket x_0 \rrbracket = \llbracket n_0 \rrbracket \land \llbracket x_1 \rrbracket = \llbracket n_1 \rrbracket) \}$

$P(N_1)' = \{ c_0!m_0@N_0 \triangleleft c_0!m_1@N_0, (c_1!n_0@N_1 \triangleleft c_1!n_1@N_1) * (c_0?y_0@N_1 \triangleleft c_0?y_1@N_1 \land \llbracket y_0 \rrbracket = \llbracket m_0 \rrbracket \land \llbracket y_1 \rrbracket = \llbracket m_1 \rrbracket) \}$.

Now, we specify the network in Table~\ref{tab:network_spec} where $\mathfrak{C}@N_i$ denotes programs running on $N_i$. It is trivial to give the proof with the combination of system proofs by \textit{Nest Composition Rule}.

\begin{table*}
\caption{Specification of the network in Example~\ref{eg:nest}.}
\label{tab:network_spec}
\centering
  \begin{tabular}{ccc}
  & $\{ \top, m_0, m_1 \mapsto -, - * x_0, x_1 \mapsto -, - *$ & \\
  & $n_0, n_1 \mapsto -, - *y_0, y_1 \mapsto -, - \}$ &\\
  $\{ \top, m_0, m_1 \mapsto -, - *$ & & $\{ \top, n_0, n_1 \mapsto -, - *$ \\
  $x_0, x_1 \mapsto -, - \}$ & & $y_0, y_1 \mapsto -, - \}$ \\
  $\mathfrak{C}@N_0$ & $\parallel_0$ & $\mathfrak{C}@N_1$ \\
  $P(N_0)'$ & & $P(N_1)'$ \\
  \multicolumn{3}{c}{$\{\top, (c_0!m_0 \triangleleft c_0!m_1) * (c_1?x_0 \triangleleft c_1?x_1 \land \llbracket x_0 \rrbracket = \llbracket n_0 \rrbracket \land \llbracket x_1 \rrbracket = \llbracket n_1 \rrbracket) *$} \\
  \multicolumn{3}{c}{$(c_1!n_0 \triangleleft c_1!n_1) * (c_0?y_0 \triangleleft c_0?y_1 \land \llbracket y_0 \rrbracket = \llbracket m_0 \rrbracket \land \llbracket y_1 \rrbracket = \llbracket m_1 \rrbracket) \}$}
  \end{tabular}
\end{table*}

\end{example}

\section{Language}
The full expressiveness of the ECSL can be exploited when specifying and verifying a complex system that has a nested structure and cross-layer communications. Firstly, we define a programming language with the built-in support of communication actions and nested parallel construction. Furthermore, we define a specification language to formalize the systems constructed by the programming language.

\subsection{Programming Language}
We define a programming language that has powerful expressiveness to construct systems at different abstraction levels.

\begin{definition}[Syntax of Programming Language]
The syntax of the programming language is defined as follow:

$\bar{E} ::= x ~|~ n ~|~ \bar{E}+\bar{E} ~|~ \bar{E}-\bar{E} ~|~ ...$

$\bar{B} ::= \top ~|~ \bot ~|~ \bar{B} \land \bar{B} ~|~ \bar{E}=\bar{E} ~|~ ...$

$\bar{C} ::= \textbf{skip} ~|~ x \gets \bar{E} ~|~ x \gets [\bar{E}] ~|~ [\bar{E}] \gets \bar{E} ~|~ ...$

$\bar{A} ::=  \bar{C} ~|~ \textbf{send}(\bar{E}, \textit{ch}) ~|~ x \gets \textbf{receive}(\textit{ch})$

$\bar{S} ::= \bar{A} ~|~ \bar{S_1};\bar{S_2} ~|~ \textbf{if } \bar{B} \textbf{ then } \bar{A_1} \textbf{ else } \bar{A_2} ~|~ \textbf{while } \bar{B} \textbf{ do } \bar{A}$

$\bar{P} ::= \bar{S} ~|~ \bar{P_1} \parallel \bar{P_2} ~|~ \bar{P_1} \parallel_i \bar{P_2}$
\end{definition}

We omit some usual arithmetic expressions in $\bar{E}$ and Boolean expressions in $\bar{B}$. Commands $\bar{C}$ include elementary actions such as the empty command, assignment command and memory management commands while actions $\bar{A}$ encapsulate communication commands including \textbf{send} and \textbf{receive}. Statements $\bar{S}$ defines the basic program structure. Parallel structures $\bar{P}$ describe the construction of nested systems.

\subsection{Specification Language}
To specify systems developed by the programming language above, we define a specification language.

\begin{definition}[Syntax of Specification Language of ECSL]
The assertion component of the specification language is structured in Definition~\ref{def:spec_lang_csl}. Let $\Phi$ denote a grammar in a temporal logic. The new syntax is defined as follow:

$\check{E} ::= \Phi$

$\check{Q} ::= \check{E}_f, \check{E}_n \land \check{P}$
\end{definition}

\begin{example}
We take linear temporal logic (LTL) as an example. For the syntax of LTL over the set \textit{Prop} of proposition with $Q \in \textit{Prop}$, $\Gamma$ and $\gamma$ of program $\mathfrak{P}$ in system $\mathfrak{W}$ are formed as LTL formulae according to the following grammar:

$\Phi ::= \top ~|~ Q ~|~ \neg \Phi ~|~ \Phi_0 \land \Phi_1 ~|~ \bigcirc \Phi ~|~ \Phi_0 \sqcup \Phi_1$.

Unary prefix operator $\bigcirc$ and binary infix operator $\sqcup$ are temporal modalities. If $\Phi$ holds in the next time step, $\bigcirc \Phi$ holds at the current moment. $\Phi_0 \sqcup \Phi_1$ holds at the current moment, if there is a future moment $i$ for which $\Phi_1$ holds and $\Phi_0$ holds at all moments until moment $i$.
\end{example}

With Definition~\ref{def:program_state}, we can give the operational semantics of this specification language.

\begin{definition}[Semantics of Specifications]
The semantics of specifications contains the semantics of assertions and the environment factor.

The semantics of assertions is given in Definition~\ref{def:sem_CSL} with the only difference that the temporal memory is a new component of the program state according to Definition~\ref{def:program_state}.

The semantics of the environment factor is defined as follows with Theorem~\ref{the:trace_satisfication}:

$(s, h, t) \models (\check{E}_f, \check{E}_n \land \check{P}) \iff \exists h_0, h_1, t_f, t_n: h = h_0 \uplus h_1 \land t = t_f \uplus t_n \land (s, h) \models \check{P} \land \mathcal{T}(t_f) \models \check{E}_f \land \mathcal{T}(t_n) \models \check{E}_n$

\end{definition}

\section{Discussion}
We design ECSL by following unitarity and compatibility, which makes ECSL capable of specifying and verifying systems at different abstraction levels and interpreting the standard CSL and typical variants such as \cite{bell_concurrent_2010,sergey_programming_2017}. We also expand the capability of the CSL to handle the nested architecture compared to the work \cite{krogh-jespersen_aneris_2020}. The most distinguishing feature of ECSL compared to the current work is to focus on tackling typical modularity issues in a unified manner.

Our general idea about proving the soundness of ECSL is to formulate a new semantics of judgements according to \cite{vafeiadis_concurrent_2011} considering an auxiliary predicate $\textit{sate}_n(\alpha, s, h, t, J, P')$, where $\alpha$ denotes the action executing with stack $s$, a heap $h$ and recorded in temporal memory $t$ while the judgement form is defined in Definition~\ref{def:judgement_form}. The soundness can be formulated as follows:

$J \vdash \{ \Gamma, \gamma \land P \}~\alpha~\{ \Gamma, \gamma' \land P' \} \implies J \models \{ \Gamma, \gamma \land P \}~\alpha~\{ \Gamma, \gamma' \land P' \}$.

We can prove soundness by proving each proof rule of ECSL is a sound implication after replacing all $\vdash$ by $\models$ though the proof of environment rules is challenging.

Although we do not address the representation of temporal properties in our examples, we can indeed represent them in temporal logic. For instance, the foreign environment of the pre-condition can be represented as $\Box(c!m \longrightarrow \Diamond c?v)$, meaning that whenever signal $m$ is sent through channel $c$, then $v$ will eventually be assigned with the value of signal $m$ received from channel $c$. 

\section{Conclusion}
Reasoning about complex distributed systems requires great modularity. In this paper, we have proposed ECSL, an extended concurrent separation logic with the temporal extension, communication extension, environment extension, and nest extension. We equip ECSL with spatiotemporal modularity to enrich expressiveness. Besides, ECSL facilitates formalizing systems at different abstraction levels. Particularly, ECSL is capable of formalizing nested systems containing both peer communications and cross-layer communications, which enables ECSL to specify and verify complex systems. Furthermore, ECSL presents unitarity and compatibility, which is implemented in a specification language to formalize systems at different abstraction levels.

Our next research direction is to apply mechanized ECSL into formalizing practical systems, especially decentralized systems such as the consensus algorithm, smart contract, and blockchain design.

\bibliographystyle{splncs04}
\bibliography{ecsl}

\end{document}